\documentclass[letterpaper,11pt,reqno]{amsart} 

\usepackage{times}
\usepackage{amsfonts, latexsym, amsmath, amssymb, epic, eepic, epsfig,fullpage,psfrag,url}
\usepackage[compress]{natbib}
\bibpunct{(}{)}{,}{a}{}{,}
\usepackage{graphicx}
\usepackage[dvips]{color}
\usepackage{setspace}

\newtheorem{theorem}{Theorem}[section]
\newtheorem{definition}{Definition}
\newtheorem{lemma}[theorem]{Lemma}

\newtheorem{corollary}[theorem]{Corollary}

\newtheorem{example}[theorem]{Example}
\newtheorem{remark}[theorem]{Remark}

\newcommand{\altp}{\pi}
\newcommand{\braessp}{p}
\newcommand{\braessq}{q}
\newcommand{\braessr}{r}

\newcommand{\xvec}{x}
\newcommand{\zvec}{z}

\newcommand{\std}{ \sigma }

\newcommand{\var}{ v }
\newcommand{\kappavar}{\kappa}
\newcommand{\kappastdev}{ \kappa_{\std} }

\newcommand{\route}{p}
\newcommand{\paths}{\mathcal{P}}

\newcommand{\flow}{f}

\newcommand{\R}{\mathbb{R}}

\newcommand{\pathcost}{Q}
\newcommand{\xx}{x}   
\newcommand{\zz}{z}   
\newcommand{\CC}{C} 
\newcommand{\GV}{V}
\newcommand{\GA}{E}

\definecolor{myred}{rgb}{0.95, 0.01, 0.01}
\definecolor{myblue}{rgb}{0.01, 0.01, 0.9}

\title{The Burden of Risk Aversion in Mean-Risk Selfish Routing}

\author[E. Nikolova]{E. Nikolova${}^*$}
\author[N.E. Stier-Moses]{N.E. Stier-Moses${}^{**}$}
\dedicatory{
  ${}^*$
 Department of Electrical and Computer Engineering, University of Texas at Austin, Austin, TX, USA,
  {\rm\url{nikolova@austin.utexas.edu}}
\\[1ex]
 ${}^{**}$ Graduate School of Business, Columbia University, New York, USA, {\rm\url{stier@gsb.columbia.edu}}\\
 School of Business, Universidad Torcuato Di Tella and CONICET, Buenos Aires, Argentina
}

\date{Oct 2014.}

\begin{document}
\begin{abstract}
Considering congestion games with uncertain delays, we compute the
inefficiency introduced in network routing by risk-averse agents. At
equilibrium, agents may select paths that do not minimize the expected
latency so as to obtain lower variability. A social planner, who is likely to be
more risk neutral than agents because it operates at a longer time-scale,
quantifies social cost with the total expected delay along routes. From that
perspective, agents may make suboptimal decisions that degrade long-term
quality. We define the {\em price of risk aversion} (PRA) as the worst-case
ratio of the social cost at a risk-averse Wardrop equilibrium to that where
agents are risk-neutral. For networks with general delay functions and a
single source-sink pair, we show that the PRA depends linearly on the agents'
risk tolerance and on the degree of variability present in the network. In
contrast to the {\em price of anarchy}, in general the PRA increases when the
network gets larger but it does not depend on the shape of the delay functions.
To get this result we rely on a combinatorial proof that employs alternating
paths that are reminiscent of those used in max-flow algorithms. For {\em
series-parallel} (SP) graphs, the PRA becomes independent of the network
topology and its size. As a result of independent interest, we prove that for
SP networks with deterministic delays, Wardrop equilibria {\em maximize} the
shortest-path objective among all feasible flows.
\\[1mm]

\end{abstract}

\maketitle


\section{Introduction}\label{sec:intro}

A central question in decision making is how to make good decisions under
uncertainty, in particular when decision makers are risk averse.
Applications of crucial national importance, including alleviating congestion in
transportation networks, as well as improving telecommunications, robotics, security
and others, all face pervasive uncertainty and often require finding {\em
reliable} or {\em risk-minimizing solutions}. Those applications have
motivated the development of algorithms that incorporate risk primitives, and
the inclusion of risk aversion in questions related to algorithmic game theory.
While risk has been extensively studied in the fields of finance and
operations, among others, in comparison there is relatively little literature
in the theoretical computer science community devoted to this issue. One of
the goals of this paper is to inspire more work devoted to understanding and
mitigating risk in networked systems.

Capturing uncertainty and risk aversion in traditional combinatorial
problems studied by theoretical computer science often reduces to nonlinear or
nonconvex optimization over combinatorial feasible sets, for which no
efficient algorithms are known. Possibly due to the difficulty in writing the
ensuing problems in simple terms, at present we lack a systematic
understanding of how risk considerations can be successfully incorporated into
classic combinatorial problems. Doing so would necessitate new techniques for
analyzing risk-minimizing combinatorial structures rigorously.

Within the fields of algorithms and algorithmic game theory, routing has
proved to be a pervasive source of important questions. Indeed, many
fundamental questions on risk-averse routing are still open, including several
intriguing cases where
the complexity is unknown. For example, in a network with uncertain edge
delays, what is the complexity of finding the path with highest probability
of reaching the destination within a given deadline? What is the path that
minimizes the mean delay plus the standard deviation along that path?  The
best known algorithms for both of these questions have a polynomial runtime and smoothed complexity,
but a subexponential worst-case running time
\citep{nikolova-ssp,nikolova-approx}. Hence, these problems are unlikely to
be NP-hard, yet polynomial algorithms have so far been elusive. There are
also a myriad of other possible risk objectives, yielding nonlinear and
nonconvex optimization problems over the path polytope, that are open from a
complexity, algorithmic and approximation point of view.

Consequently, we are at the very beginning of understanding risk and
designing appropriate models in routing games, which have been instrumental
in the development of algorithmic game theory in the past two decades.
Routing games capture decision-making by multiple agents in a networking
context, internalizing the congestion externalities generated by
self-minded agents. Externalities are traditionally captured by considering
that edge delays are functions of the edge flow.
In addition to the technical challenges involving the characterization and
computation of equilibria, a key insight brought by the study of these
games was that equilibria are not extremely inefficient. The {\em price of
anarchy}---by now a widely studied concept used in a variety of problems and
applications---represents the worst-case ratio between the cost of a Nash
equilibrium and the cost of the socially-optimal solution. This ratio, which
quantifies the degradation of system performance due to selfish behavior, was
first defined in the context of routing and applied to a network with
parallel links~\citep{papadimitriou-equilibriaJournal}. It was subsequently
analyzed for general networks and different types of
players~\citep{roughgarden-selfrouting,roughgarden-priceanarchy,css-capsoue,css-congestion}.
With few exceptions, this stream of work has assumed that delays are
deterministic, while almost every practical situation in which such games
could be useful presents uncertainty. For instance, there are uncertain
delays in a transportation network due to weather, accidents, traffic lights,
etc., and in telecommunications networks due to changing demand, hardware
failures, interference, packet retransmissions, etc.

\vspace{2mm}
\noindent{\bf Risk Model.}
A generalization of the classic selfish routing model introduced by
\cite{beckmann-transportation} to the case of uncertain delays is to
associate every edge to a random variable whose distribution depends on the
edge flow.
The problem faced by a risk-averse agent becomes more involved since it is
not merely finding the shortest path with respect to delays. To find the best
route, agents must consider both the expected delay and the variability along
all possible choices, leading them to solve stochastic shortest path problems.
For instance, it is common that commuters add a buffer to the expected travel
time for their trip to maximize the chance of arriving on-time to an
important meeting or to a flight at the airport. A classic model
in finance that captures the tradeoff between mean and variability is
Markowitz' mean-risk framework~\citep{Markowitz52}. It considers an agent that optimizes a linear
combination of mean and risk, weighted by a {\em risk-aversion coefficient}
that quantifies the degree of risk aversion of that agent. In the
context of routing, \cite{nikolova-stochWE} adapted that framework to
Wardrop equilibria, showed the existence of equilibria, and computed
efficiency indicators such as the price of anarchy.

Of course, there are multiple ways to capture risk. The expected utility
theory~\cite{Neumann44a}, which is prevalent in economics, captures risk-averse preferences
using concave utility functions.  
This theory has been criticized due to unrealistic assumptions such as
independence of irrelevant alternatives so other theories have been proposed (e.g., \citet{TverskyK:1981}).
The theory of coherent risk measures, proposed in the late nineties, takes an
axiomatic approach to risk (e.g., see surveys~\citet{Rockafellar07,Uryasev11} and references therein). \cite{cominetti-riskaverserouting} adapted these ideas to the
context of network routing and concluded that the mean-variance objective has
benefits over other risk measures in being additively consistent. In finance, the
mean-risk and other traditional risk measures have been criticized for
leading to counterintuitive solutions such as preferring stochastically
dominated solutions.
Nevertheless, different risk postulates may be relevant in the context of
transportation and telecommunications.  For instance,
stochastically-dominated routes may be admissible if certainty is more valued
than a stochastically-dominant solution with large variance. Indeed, one may
choose a larger-latency path rather than routing along variable paths that
introduce jitter in real-time communications.
Following on our previous work on risk-averse congestion games, in this paper
we consider both the mean-variance and mean-standard deviation objectives for
risk-averse routing.

It is important to note that risk aversion may induce agents to choose
longer routes to reduce risk, effectively trading off mean with variance.
Hence, a natural question to ask for a network game with uncertain delays
and risk-averse agents is {\em how much of the degradation in system
performance can be attributed to the agents' risk-aversion}. We refer to this
degradation by the {\em price of risk-aversion} (PRA), which we formally
define as the worst-case ratio of the social cost of the equilibrium (with
risk-averse agents) to that of an equilibrium if agents were risk neutral. The
rationale for choosing this particular ratio is that we want to disentangle
the effects caused by selfish behavior, captured by the price of anarchy,
from those caused by risk aversion per se.
The social cost is considered with respect to average delays
because a central planner would typically care about a long-term perspective
and minimize average agent delays and average pollutant emissions.




Using the variance of delay along a route as a risk indicator, although not
directly intuitive since it is not expressed in the same units as the mean
delay, leads to models that satisfy natural and intuitive optimality
conditions for routes; namely, a subpath of an optimal path remains optimal
(called the {\em additive consistency} property).
Indeed, the mean-variance objective is additive along paths (the cost of a path is the
sum of the cost of its edges).
It thus lends itself to
tractable algorithms in terms of computing equilibria, at least as long as
delays are pairwise independent across edges.
On the other hand, comparing a Wardrop equilibrium with risk-averse players
to a standard Wardrop equilibrium is far from straightforward and
requires new techniques for understanding how the two differ.

Alternatively, one could set the risk indicator to be the standard deviation of
delays. A big advantage is that the mean-standard deviation objective can
be thought of as a quantile of delay, easily justifying the buffer time that
commuters consider when selecting when to start the trip. The disadvantages
are that additive consistency is lost and, technically, that to compute the
standard deviation one must take a square root, which makes the objective
non-separable and nonconvex. For more details, we refer the reader to
\cite{nikolova-stochWE}, where
these pros and cons are discussed in further detail.

\subsection*{Our results}

We define a new concept, the price of risk aversion (PRA), as
the worst-case ratio of the social cost (total expected delay) of a risk-averse
Wardrop equilibrium (RAWE) to that of a risk-neutral Wardrop equilibrium
(RNWE).
Our main result, presented in Section~\ref{sec:general}, is a bound on the
price of risk aversion for arbitrary graphs with a single origin-destination (OD)~pair and mean-variance
risk minimizing players. We provide a bound of $1+\gamma\kappa\eta$, where $\gamma$ is the
risk-aversion coefficient, $\kappa$ is the maximum possible variability (variance to mean ratio) of all
edges when the prevailing traffic conditions are those under the equilibrium,
and $\eta$ is a topological parameter that captures how many flow-bearing paths are
needed to cover a special structure called an alternating path.
The parameter $\eta$ strongly depends on the
topology, and is at most half the number of nodes in the network, $\lceil (n-1)/2 \rceil$. The
resulting bound is appealing in that it depends on the three factors that one
would have expected (risk aversion, variability and network size), but perhaps
unexpectedly does so in a linear way and for arbitrary delay functions. The proof of
this result is based on a novel idea of constructing a type of {\em alternating
path} that switches between forward edges for which the flow under a RAWE is
less than or equal to the flow under a RNWE, and backward edges for which the
opposite inequality holds. This construction is the key that allows us to
compare both equilibrium flows and derive our main result.

The proof of the main result consists of three key lemmas that show that (a)
an alternating
path always exists (Lemma~\ref{lemma0}), (b) the cost of a RAWE is
upper-bounded by an inflated total mean delay along forward edges minus the
total mean delay along backward edges (Lemma~\ref{lemma1}), and (c) the cost of a
RNWE is lower-bounded by the total mean delay along forward edges
minus the total mean delay along backward edges (Lemma~\ref{lemma2}). Steps
(a) and (c) are proved independently of the choice of risk model.
Step (b) is more subtle: it constructs a series of subpaths that connect different parts of the
alternating path to the source and the sink, and uses the equilibrium
conditions to provide partial bounds for subpaths of the alternating path.
The lemma then exploits the linearity of the mean-variance objective to get a
telescopic sum that simplifies precisely to the total delay along the
alternating path.

Theorem~\ref{thm:general} puts the lemmas together and upper bounds the total
mean delay of the forward subpaths in the alternating path by the cost of the
RNWE times the number of such forward paths, obtaining the factor $\eta \leq \lceil (n-1)/2 \rceil$ in the
worst-case, as mentioned above. Although it is an open question whether our bound is tight
for general graphs, we prove that it is tight for two
families of graphs. For series-parallel (SP) graphs, it turns out that there
must exist an alternating path that consists of only forward edges (that is, $\eta = 1$), which
implies that the price of risk-aversion for those topologies is exactly
$1+\gamma\kappa$ (Corollary~\ref{cor:sp} and lower bound in
Section~\ref{sec:lowerbound}). For Braess graphs, we establish that the price
of risk aversion is bounded by $1+2\gamma\kappa$ and this bound is tight, as well.

As mentioned above, many of the results for the mean-variance risk model
extend to the mean-stdev objective. In particular, the only piece missing to prove
a general theorem is an equivalent of Lemma~\ref{lemma1}, which bounds the
cost of the RAWE by an expression of the edge delays along the alternating
path.  The difficulty in extending our current proof to general graphs is the
nonlinearity of the mean-stdev cost function, which in turn puts a
restriction on the equilibrium flow in that its edge-flow representation cannot
be decomposed arbitrarily to a path-flow representation. (The latter leads to
an interesting open problem, posed by \cite{nikolova-stochWE}: is
there an efficient algorithm that converts a given equilibrium edge-flow
vector into an equilibrium path-flow decomposition? That reference shows that
a succinct path flow decomposition that uses polynomially-many paths exists.)
Circumventing the nonlinearity challenge, we are able to prove the equivalent
of Lemma~\ref{lemma1} on the Braess graph with a more involved case analysis
(Lemma~\ref{lemma1:braess-stdev}). Henceforth, we establish that the exact
value of the price of risk aversion for Braess graphs in the
mean-stdev case match those in the mean-variance case.

The independence of the network topology property for SP graphs also extends to
the mean-stdev case. To obtain
that extension,
we provide a result for SP graphs that is interesting in its own
right. As discussed earlier, it has already been established that RNWE are
typically not extremely inefficient because the price of anarchy is bounded.
We prove that considering SP graphs with deterministic delays, the equilibrium
maximizes the shortest path objective among all feasible flows
(Theorem~\ref{thm:SP}).
In other words, the shortest
path for an arbitrary flow can never be longer than the shortest path for the
RNWE.

\section{Related work}\label{sec:related}


In this work we consider how having stochastic delays and risk-averse users influence the
traditional competitive network game introduced by \cite{wardrop-traffic}. He
postulated that the prevailing traffic conditions can be determined from the
assumption that users jointly select shortest routes, and the mathematics
that go with this idea were formalized in an influential book by
\cite{beckmann-transportation} where they lay out the foundations to analyze
network games. These models find applications in various application domains
such as in transportation \citep{sheffi-book} and telecommunications
\citep{altman-surveynetgames}.
In the last decade, these types of models have received renewed attention
with many studies aimed at understanding under what conditions these games
admit an equilibrium, what uniqueness properties are satisfied by these
equilibria, what methods can be used to compute equilibria efficiently, what
price is paid for having competition instead of a centralized solution, and
what are good ways to align incentives so an equilibrium becomes socially
optimal. For references on these topics from a perspective similar to ours,
we refer the readers to the surveys \citep{nisan-algoGT,cs-wardrop}.

The route-choice model in this paper consists of users that select the path
that minimizes the mean plus a multiple of the variability of travel time
(captured by either variance or standard deviation).
Exact algorithms and fully-polynomial approximation schemes have been proposed for this problem and for a more general
risk-averse combinatorial framework by \cite{nikolova-ssp,nikolova-approx}.
Approximation algorithms for other risk-averse combinatorial frameworks
were provided by~\cite{brand-shochSP,Swamy:2011,LiD:2011}.  We refer the reader to these papers for
a more extensive list of references on risk-averse combinatorial optimization.





There is a growing literature on stochastic congestion games. \cite{os-rweTS}
introduce a game with uncertain delays and risk-averse users and study
how the solutions provided by it can be approximated numerically by an
efficient column-generation method that is based on robust optimization. The
main conclusion is that the solutions computed using their approach are good
approximations of {\em percentile equilibria} in practice. Here, a percentile
equilibrium is a solution in which percentiles of travel times along
flow-bearing paths are minimal. They also use their algorithm to compare
equilibria with risk-averse players to those with risk-neutral players, as in
the standard Wardrop model. \cite{nie-perceq}, who also studies percentile
equilibria, considers an instance with two edges and exogenous standard
deviations in detail, provides a gradient projection algorithm to find
percentile equilibria, and uses it to perform a computational study.
\cite{nikolova-stochWE} prove existence and POA results for the exact model
we consider here, when the variability is captured by the standard deviations
of latencies.
\cite{PiliourasNS:2013} consider the sensitivity of the price of anarchy to several risk averse user objectives, in a different routing game model with atomic players and affine latency functions.
\cite{AngelidakisFL13} also focus on atomic congestion games with uncertainty induced by stochastic players or stochastic delays, and characterize when equilibria can be efficiently computed.
For other references, we direct the reader to \cite{nikolova-stochWE}.

For network and congestion games, a series of papers in the last 15 years
have studied the inefficiency introduced by self-minded behavior. To quantify
that inefficiency, \cite{papadimitriou-equilibriaJournal} computed the
supremum over all problem instances of the ratio of the equilibrium cost to
the social optimum cost, which has been called the {\em price of anarchy}
(POA) by \cite{papadimitriou-algogamesinternet}. The POA has been analyzed extensively
in relation to transportation and telecommunications networks with
increasingly more general assumptions
\citep{roughgarden-selfrouting,roughgarden-priceanarchy,css-capsoue,chau-soue,perakis-soue,css-congestion}.
\cite{nikolova-stochWE} extended that notion to the case of stochastic
delays with risk-averse players.  A different concept, the price of uncertainty,
was considered in congestion games in reference to how best response dynamics change
under randomness introduced by an adversary and random ordering of players~\citep{BalcanBM:2009}.
%
%
Risk aversion in the algorithmic game theory literature has been considered recently in the context of general games (e.g., \citet{FiatP:2010} who focus on hardness results) and mechanism design (e.g., \citet{DughmiP2012,FuHH:2013,Dughmi2014}).



\section{The Model}\label{sec:model}

We consider a directed graph $G=(\GV,\GA)$ with a single source-sink pair
$(s,t)$ and an aggregate demand of $d$ units of flow that need to be routed
from $s$ to $t$. We let $\paths$ be the set of all feasible paths between
$s$ and $t$. We encode players decisions as a flow vector $\flow
=(\flow_{\route})_{\route\in\paths}\in \R^{|\paths|}_+$ over all paths. Such
a flow is feasible when demand is satisfied, as given by the constraint
$\sum_{\route\in\paths} \flow_{\route}=d$. For notational simplicity, we
denote the flow on a directed edge $e$ by $\flow_e = \sum_{\route \ni e}
\flow_{\route}$. When we need multiple flow variables, we use the analogous
notation $\xvec, x_{\route}, x_e$ and $\zvec, z_{\route}, z_e$.

The network is subject to congestion, modeled with stochastic delay functions
$\ell_e(\flow_e) +\xi_e(\flow_e)$ for each edge $e\in \GA$. Here,
$\ell_e(\flow_e)$ measures the expected delay when the edge has flow $\flow_e$,
and $\xi_e(\flow_e)$ is a random variable that represents a noise term on the
delay, encoding the error that $\ell_e(\cdot)$ makes. Functions
$\ell_e(\cdot)$, generally referred to as {\em latency functions}, are assumed
continuous and non-decreasing. The expected latency along a path $p$ is given by
$\ell_{\route}(f):=\sum_{e\in \route} \ell_e(\flow_e)$.

Random variables $\xi_e(\flow_e)$ have expectation equal to zero and standard
deviation equal to $\std_e(\flow_e)$, for arbitrary continuous 
functions $\std_e(\cdot)$.
%
%
We assume that these random variables are pairwise independent. From there,
the variance along a path equals $\var_{\route}(f)=\sum_{e\in
\route} \std_e^2(\flow_e)$, and the standard deviation (stdev) is
$\sigma_{\route}(f)=(\var_{\route}(f))^{1/2}$.
We will initially work with variances and then extend the model to standard
deviations, which have the complicating square roots. (For details on the
complications, we refer
the reader to \citet{nikolova-stochWE}).

We will consider the {\em nonatomic} version of the routing game where
infinitely many players control an infinitesimal amount of flow each so that
the path choice of a single player does not unilaterally affect the costs
experienced by other players (even though the joint actions of players affect
other players).

As explained in the introduction, players are risk-averse and hence choose paths taking into account the variability of delays.
We follow the literature and perturb the mean delay of path $p$ with a factor of the variance:
\begin{equation}\label{eqn:pathcost-var}
\pathcost^{\gamma}_{\route}(\flow) =
\ell_{\route}(\flow)+\gamma \var_{\route}(\flow).
\end{equation}
This objective function will be referred to as the {\em mean-var} objective, and
frequently simply as the {\em path cost} (as opposed to latency). Here, $\gamma\ge 0$
is a constant that quantifies the risk-aversion of the players, which we
assume homogeneous. The special case of $\gamma = 0$ corresponds to
risk-neutrality.

The variability of delays is usually not too large with respect to the
expected latency. It is common to consider the {\em coefficient of
variability} $CV_e(\flow_e):=\std_e(\flow_e)/\ell(\flow_e)$ given by the
ratio of the standard deviation to the expectation as a relative measure of
variability. In this case, we consider the {\em variance-to-mean ratio}
$\var_e(\flow_e)/\ell(\flow_e)$ as a relative measure of variability.
Consequently, we assume that $\var_e(\xx_e)/\ell_e(\xx_e)$ is bounded from
above by a fixed constant $\kappavar$ for all $e\in \GA$ at the equilibrium
flow $\xx_e\in \R_+$. This means that the variance cannot be larger than
$\kappavar$ times the expected latency in any edge at the equilibrium flow.

%
%

In summary, an instance of the problem in given by the tuple
$(G,d,\ell,\var,\gamma)$, which represents the topology, the demand, the
latency functions, the variability functions, and the degree of risk-aversion of players.



The following definition captures that at equilibrium players route flow
along paths with minimum cost $\pathcost^{\gamma}_{\route}(\cdot)$. In
essence, users will switch routes until at equilibrium costs are equal along
all used paths.  This is the natural extension of the traditional Wardrop
Equilibrium to risk-averse users.

\begin{definition}[Equilibrium]\label{defi:eq}
A {\em $\gamma$-equilibrium} of a stochastic nonatomic routing game is a flow
$\flow$ such that for every $k\in K$ and for every path $p\in\paths_k$
with positive flow, the path cost $\pathcost^{\gamma}_{p}(\flow) \leq
\pathcost^{\gamma}_{q}(\flow)$ for any other path $q \in \paths_k\,$. For a
fixed risk-aversion parameter $\gamma$, we refer to a $\gamma$-equilibrium as a {\em risk-averse
Wardrop equilibrium} (RAWE) and denote it by $\xx$.
\end{definition}

Notice that since the variance decomposes as a sum over all the edges that
form the path, the previous definition represents a standard Wardrop
equilibrium with respect to modified costs $\ell_e(\flow_e)+\gamma\var_e(\flow_e)$.
For the existence of the equilibrium, it is sufficient that the modified cost functions are increasing. 

Our goal is to investigate the effect that risk-averse decision-makers have
on the quality of equilibria. The quality of a solution that represents
collective decisions can be quantified by the cost of equilibria with respect
to expected delays since, over time, different realizations of delays average
out to the mean by the law of large numbers. For this reason, a social
planner, who is concerned about the long term, is typically risk neutral, as
opposed to users who tend to be more emotional about decisions. Furthermore,
the social planner may aim to reduce long-term emissions, which would be
better captured by the total expected delay of all users. These arguments
justify the difference between the risk aversion coefficient that
characterizes user behavior at equilibrium and the behavior of the social
planner.

\begin{definition}
The {\em social cost} of a flow $\flow$ is
defined as the sum of the expected latencies of all players:
\begin{equation}\label{eq:soccost}
\CC(\flow) := \sum_{\route \in \paths} \flow_{\route}\ell_{\route}(\flow)=
\sum_{e \in \GA} \flow_e\ell_e(\flow_e).
\end{equation}
\end{definition}

Although one could have measured total cost as the weighted sum of the costs
$\pathcost^{\gamma}_{\route}(\flow)$ of all users, this captures users'
utilities but not the system's benefit. Although one can consider such a cost
function to compute the price of anarchy (see \citet{nikolova-stochWE}), in the
current paper our goal is to compare across different values of risk aversion
so we want the various flow costs to be measured in the same unit.

The following definition captures the increase in social cost at equilibrium
introduced by user risk-aversion, compared to the cost one would have if users
were risk-neutral. Hence, we use a {\em risk-neutral Wardrop equilibrium}
(RNWE), defined as a $0$-equilibrium according to Definition~\ref{defi:eq}, as
the yardstick to determine the inefficiency caused by risk-aversion. We define
the price of risk aversion as the worst-case ratio among all possible instances
of expected costs of the risk-averse and risk-neutral equilibria.

\begin{definition}
Consider a family of instances $\mathcal{F}$ of a routing game with uncertain
delays. The {\em price of risk aversion} (PRA) associated with the
risk-aversion coefficient $\gamma$ is defined by
\begin{equation}\label{eq:PoR}
PRA(\mathcal{F},\gamma):=\sup_{
G,d,\ell,\var:
(G,d,\ell,\var,\gamma)\in \mathcal{F}
}\qquad\frac{\CC(\xx)}{\CC(\zz)},
\end{equation}
where $\xx$ and $\zz$ are the RAWE and the RNWE of the
corresponding instance.
\end{definition}

Intuitively, this ratio would depend on the topology of the instance $G$, on
the risk-aversion coefficient $\gamma$, and on the degree of the variability of
delays $\kappavar$. We present the following example to motivate the form of
the bound to the PRA, which is going to be linear both in $\gamma$ and in
$\kappavar$.
%
%
The example is based on a simple network with two edges, usually referred to as the {\em Pigou network}
\citep{pigou1920,roughgarden-selfrouting}.

\begin{example}\label{ex:pigou}
Consider an instance with two nodes connected by two parallel edges with
latencies equal to $(1+\gamma\kappavar)x$ and $1$, respectively, and
variances equal to $\var_1=0$ and $\var_2=\kappavar$. The total demand is a
unit. Computing equilibria, the RNWE flow routes $1/(1+\gamma\kappavar)$
units along the first edge and $\gamma\kappavar/(1+\gamma\kappavar)$ along
the second. This gives a total cost of $1$. Instead, the RAWE flow routes all
the flow along the first edge, which gives a total cost of
$1+\gamma\kappavar$. Dividing, we get that PRA$\ge 1+\gamma\kappavar$.
\end{example}

The previous example motivates the need of imposing an upper bound on the
variability of delays.

\begin{remark}\label{remark:pigou}
Taking $\kappavar\to\infty$ in the previous example, it follows that if one
does not constrain variability of delays, the price of risk aversion is
unbounded.
\end{remark}

Having bounded variability is a
reasonable assumption in real-life transportation networks since the
variability is never too many times larger than the expected latency of an
edge. Moreover, the more congested the network is, the less variable delays
are since speeds approach zero and hence the possibilities of variation are
minimal. In the following section, we shall prove that $1+\gamma\kappavar$ is
a matching upper bound for instances with the topology of Pigou networks.
Indeed, we will see that this will be a special case of a result for general
topologies.

\section{PRA in General graphs}\label{sec:general}

We first generalize the lower bound given in the previous section to suggest
what one can expect in general and then find an upper bound to the price of risk aversion (PRA).

\subsection{A lower bound on the PRA for the Braess graph}\label{sec:lowerbound}

We exhibit a family of instances with the topology of a Braess network
\citep{braess-paradox,roughgarden-selfrouting} where the PRA tends to
$1+2\gamma\kappavar$, establishing that lower bound.

\begin{example}\label{ex:braessVar}
Consider the symmetric Braess network instance in
Figure~\ref{fig:pra-counterexample}. We refer to the top path as $\braessp$, which
consists of edges $(a,b)$. Similarly, we refer to the bottom path as $\braessq$,
which consists of edges $(c,d)$. Last, we refer to the zigzag path as $\braessr$,
which consists of edges $(a,e,d)$. Latencies are given by $\ell_a(x) =
\ell_d(x) = \alpha x$, $\ell_b(x) = \ell_c(x) = 1$ and $\ell_e(x) = \beta + v$,
and edge variances by $\var_a(x)=\var_d(x)=\var_e(x)=0$ and $\var_b(x) = \var_c(x) =
v$. Further, we set $\gamma = 1$ and $\kappavar = v$ so that $\gamma \kappavar=
v$. These definitions verify that the variance-to-mean ratio is bounded by
$\kappavar$.
\begin{figure}
\begin{center}
\psfrag{1}[c][c]{\footnotesize$(1,v)$}
\psfrag{z}[l][l]{\footnotesize$(\beta+v,0)$}
\psfrag{d}[c][c]{\footnotesize$1$}
\psfrag{x}[c][c]{\footnotesize$(\alpha x,0)$}
\includegraphics[height=30mm]{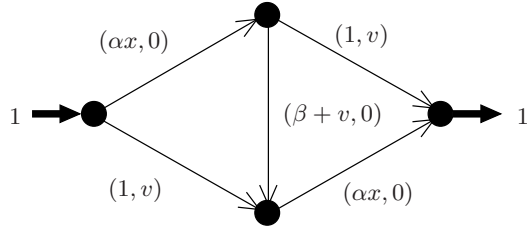}
\end{center}
\caption{Braess instance corresponding to Example~\ref{ex:braessVar}. Edges are labeled with (mean, variance) pairs.}\label{fig:pra-counterexample}
\end{figure}
%
%

\begin{itemize}
\item First, we compute the {\em risk-averse equilibrium $\xx$}. To do that, we set all
path costs to be equal at equilibrium. By symmetry, it is enough that
$\pathcost^{\gamma}_{\braessp}(\xx)= \pathcost^{\gamma}_{\braessr}(\xx)$, which is equivalent
to:
%
$$ \ell_a(\xx_a) + \ell_b(\xx_b) + \gamma \var_\braessp = \ell_a(\xx_a) + \ell_e(\xx_e) + \ell_d(\xx_d)
\Leftrightarrow  \ell_b(\xx_b) + \gamma \var_\braessp = \ell_a(\xx_a) + \ell_e(\xx_e).$$
Using the structure of Braess, $\xx_a=1-\xx_\braessp$, thus from above $1 + v = \alpha(1-\xx_\braessp) + \beta + v$, which simplifies to $1 + \alpha \xx_\braessp = \alpha + \beta$.
We will ensure that $\xx_\braessr=1$ (that is, the zigzag path carries all flow) is an equilibrium.
In that case, $\alpha + \beta = 1$, so we set $\beta = 1 - \alpha$.
Therefore, the social cost of the risk-averse equilibrium is
$\CC(\xx) = \ell_\braessr(\xx)= 2\alpha +\beta +v = \alpha + v + 1$.

\item Next, let us compute the {\em risk-neutral equilibrium $\zz$}.
%
Proceeding as before, we need that
$\ell_{\braessp}(\zz)= \ell_{\braessr}(\zz)$, which is equivalent
to:
$$\ell_a(\zz_a) + \ell_b(\zz_b) = \ell_a(\zz_a) + \ell_e(\zz_e) + \ell_d(\zz_a)
\Leftrightarrow 1 = \beta + v + \alpha(1-\zz_\braessp).$$
After some algebra, we get that
$\zz_\braessp = v/\alpha$. Since $\zz_\braessp \in [0,1/2]$ by construction,
$\alpha \geq 2 v$.
The social cost of the risk-neutral equilibrium is
$
\CC(\zz) = \ell_\braessp(\zz) =
\alpha (1-v/\alpha) + 1=
\alpha - v + 1
$.

\item Putting it all together, the {\em PRA} is
$$
\frac{\CC(\xx)}{\CC(\zz)} = \frac{\alpha+1+v}{\alpha+1-v}\,.
$$
To see how this ratio depends on $\gamma\kappavar$, we set it equal to $1+h
\gamma\kappavar= 1+h v$ and solve for $h$. Indeed,
$$
h v = \frac{\alpha+1+v}{\alpha+1-v} - 1 =
\frac{2v}{\alpha+1-v}\,,$$
from where $h = 2/(\alpha+1-v)$. To conclude that a lower bound for the PRA for
the Braess instance is $1+2\gamma\kappavar$, we see that $h\to 2$ as
$v\rightarrow 0$ when $\alpha=2v$.
\end{itemize}
\end{example}

Summarizing the previous example, at the end the worst-case instance is
parameterized only by $v$. Latencies are given by $\ell_a(x) = \ell_d(x) = 2v
x$, $\ell_b(x)=\ell_c(x)=1$ and $\ell_e(x) = 1-v$. The RNWE is
$\zz_\braessp=\zz_\braessq=1/2$ with cost $\CC(\zz)=\ell_\braessp(\zz)=1+v$, and the RAWE is $\xx_\braessr=1$
with cost $\CC(\xx)=\ell_\braessr(\xx)=1+3v$. Dividing, we get that the PRA is lower
bounded by $1+2v/(v+1)$. Hence, the worst-case bound of $1+h\gamma\kappavar$ is
achieved when $v=\kappavar\to 0$ and in that case $h\to 2$.

Having presented this lower bound on the PRA, we next derive an upper bound for general
networks and in the process we show that it is tight for Braess instances.

\subsection{An upper bound on the PRA for general graphs}\label{sec:upperbound}

We start by introducing bounds on the latency of the risk-averse Wardrop equilibrium (RAWE),
which we will use to
find the PRA. As before, we let $\zz$ denote the risk-neutral Wardrop equilibrium (RNWE) and let $\xx$ denote the
RAWE. It is well-known that, by definition, the social cost $\CC(\zz)$ of a
RNWE can be upper-bounded by the latency $\ell_\route(\zz)$ of an arbitrary path $\route\in
\paths$, and the bound is tight if the path carries flow. We now extend that
argument to a RAWE. We prove that its social cost is bounded by the cost
$\pathcost_\route(\zz)$ of an arbitrary path $\route\in\paths$. As a corollary,
$\CC(\zz)$ is also bounded by the expected latency of an arbitrary path, blown
up by a constant that depends on the risk-aversion coefficient $\gamma$ and the
maximum coefficient of variation $\kappavar$.

\begin{lemma}\label{lemma-UBcostRAWE-var}
Consider an arbitrary network instance with general latencies and variance functions.
Letting $\route\in\paths$ denote an arbitrary path (potentially not carrying flow
at equilibrium), the social cost of a RAWE
$\CC(\xx)$ is upper bounded by the path cost $\pathcost_\route(\xx)$.
\end{lemma}
\begin{proof}
From the equilibrium conditions, we have that $\ell_q(\xx) + \gamma
\var_q(\xx) \le \ell_\route(\xx) + \gamma \var_\route(\xx)$
for all paths $q\in\paths$ that carry positive flow. Therefore,
\begin{align*}
\CC(\xx) = \sum_{q\in \paths}\xx_q\ell_q(\xx)
\le& \sum_{q\in \paths}\xx_q \left[\ell_\route(\xx) + \gamma\var_\route(\xx) - \gamma\var_q(\xx)\right] \\
=& \ell_\route(\xx) + \gamma\var_\route(\xx) - \gamma\sum_{q\in \paths} \xx_q \var_q(\xx) \\
\leq& \ell_\route(\xx) + \gamma \var_\route(\xx)  = \pathcost_\route(\xx).
\end{align*}
Here, we have used the equilibrium condition, and removed a negative term.
\end{proof}

\begin{corollary}\label{cor-UBcostRAWE-var}
Consider a general instance with a single source-sink pair, general latencies
and general variance functions that satisfy that the variance-to-mean ratio at
equilibrium is bounded by $\kappavar$. Letting $\route\in\paths$ denote an arbitrary
path (potentially not carrying flow at equilibrium), the social cost
$\CC(\xx)$ of a RAWE $\xx$ is upper bounded by $(1+\gamma\kappavar)\ell_{\route}(\xx)$.
\end{corollary}
\begin{proof}
The result follows from Lemma~\ref{lemma-UBcostRAWE-var} by noting that the
mean-var cost of a path is bounded as follows:
$$
\pathcost_\route(\xx) = \ell_\route(\xx) + \gamma  { \sum_{e \in \route}\var_e(\xx_e) }
\leq \ell_\route(\xx) + \gamma  \sum_{e \in \route}\kappavar\ell_e(\xx_e)
\leq \ell_\route(\xx) (1+\gamma \kappavar),
$$
by the assumption that $\var_e(\xx_e) \leq \kappavar \ell_e(\xx_e)$ for all edges $e\in\GA$ at the equilibrium $\xx$.
\end{proof}

To get the tightest possible upper bound, one would consider the
shortest path with respect to the variances induced by the RAWE.
Selecting a specific path, we can get rid of the $\gamma\kappavar$ factor and
bound the total cost by the expected latency of that particular path, as established in the next lemma.

\begin{lemma}\label{thm-UBcostRAWE2-var}
Consider a general instance with a single source-sink pair, general latencies
and general variance functions. Letting $\route\in\paths$ denote the path that
minimizes the variance under a RAWE $\xx$ (where the path $\route$ may or may not
carry flow), the social cost $\CC(\xx)$ is upper bounded by $\ell_{\route}(\xx)$.
\end{lemma}
\begin{proof}
Let the path $\route:=\arg\min\{\var_q:q\in\paths\}$ be the
path with least variability among those used in a risk-averse equilibrium.
Using the first inequality of Lemma~\ref{lemma-UBcostRAWE-var} and that
$\var_\route(\xx)\le \var_q(\xx)$ for all $q\in\paths$, we have that $\CC(\xx)\le
\ell_\route(\xx)+\gamma\sum_{q\in\paths} \xx_q(\var_\route(\xx)-\var_q(\xx))\le
\ell_\route(\xx)$.
\end{proof}

We proceed to bound the price of risk aversion on a general graph by an
appropriate construction of an alternating path that contains edges from the
following two sets, which form a partition of the edges in $\GA$:
\begin{align*}
A =& \{e\in \GA \,|\, \zz_e \geq \xx_e \mbox{ and } \zz_e > 0 \} \\
B =& \{e\in \GA \,|\, \zz_e <\xx_e \mbox{ or } \zz_e=\xx_e=0\}
\end{align*}

We will assume from here on that edges in $B$ satisfy $\zz_e <\xx_e$ since edges that carry no flow in both equilibria ($\zz_e=\xx_e=0$) can be removed from the graph without loss of generality.
If there is a full $s$-$t$ path $\altp$ contained in the set $A$, then it is not
too hard to prove that $\CC(\xx) \le (1+\gamma\kappavar) \CC(\zz)$. In other words, this would give the lowest possible PRA bound of $1+\gamma\kappavar$. We will now
prove that this bound can be extended to {\em alternating paths} in $G$, which
are paths from $s$ to $t$ consisting of edges in $A$ plus reversed edges in $B$.
%
%
We shall refer to edges on the alternating path that belong to $A$ as forward edges and those in $B$ as backward edges.

\begin{figure}[t]
\centerline{
\psfrag{s}{\small$s$}
\psfrag{t}{\small$t$}
\psfrag{ck}{\small$C_k$}
\psfrag{ckmo}{\small$C_{k-1}$}
\psfrag{ckmt}{\small$C_{k-2}$}
\psfrag{ak}{\small$A_k$}
\psfrag{akmo}{\small$A_{k-1}$}
\psfrag{akpo}{\small$A_{k+1}$}
\psfrag{bk}{\small$B_k$}
\psfrag{bkmo}{\small$B_{k-1}$}
\psfrag{dk}{\small$D_k$}
\psfrag{dkmo}{\small$D_{k-1}$}
\psfrag{dkpo}{\small$D_{k+1}$}
\includegraphics[width=3.5in]{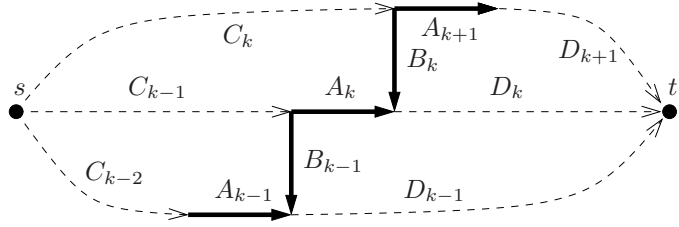}
}
\caption{Part of an alternating path. Labels denote the names of subpaths used in this section.}
\label{fig:alternatingpath}
\end{figure}

\begin{definition}
We say that a path $\altp:=(e_1,\ldots,e_r)$ from $s$ to $t$ is {\em alternating}
if reversing the direction of edges in $B$ makes it a feasible $s$-$t$ path.
\end{definition}

Figure~\ref{fig:alternatingpath} provides an illustration of the alternating
path definition (in bold) where reversing edges in $B$ creates a feasible path. The
following existence proof follows from flow conservation and the definitions of
sets $A$ and $B$.

\begin{lemma}\label{lemma0}
An alternating path exists.
\end{lemma}
\begin{proof}
We give a constructive proof that such a path exists by showing that any
$s$-$t$ cut in $G$ must have either a forward edge in $A$ or a backward edge in
$B$.  


We consider any $s$-$t$ cut defined by $S\subset \GV$ with $s \in S$ and prove that we can cross it with an edge in $A$ or a reverse edge in $B$.
%
%
Suppose the contrary, namely that
all edges incoming to $S$ are in $A$ and all edges outgoing from $S$ are in
$B$.  Denote by $\xx_A$ and $\zz_A$ the total incoming flow into $S$
corresponding to flow vectors $\xx$ and $\zz$, respectively, and by $\xx_B$ and
$\zz_B$ the total outgoing flows from $S$ respectively.  The definition of set
$A$ implies that $\xx_A \leq \zz_A$. Since conservation of flow imposes that
$\xx_B-\xx_A=\zz_B-\zz_A$, we have $\xx_B \leq \zz_B$.  On the other hand, from
the definition of $B$, either $\xx_B > \zz_B$ (note, we removed edges with $\xx_e = \zz_e = 0$),
which is a contradiction.

Now, starting with the cut $(S, G\backslash S)$ where $S = \{s\}$, we find an
appropriate edge crossing the cut and move both of its endpoints to $S$.
Thus, we add nodes to $S$ one by one until all nodes are in $S$, at
which point we will have a tree of forward and backward edges containing all
nodes. This tree yields an alternating path from the source to the
destination.
\end{proof}

We use the alternating path to provide an upper bound on the PRA that depends
on the number of times the alternating path switches from $A$ to $B$. To get
there, we need two lemmas. The first lemma extends
Corollary~\ref{cor-UBcostRAWE-var}, which applies to (standard) paths, to the
case of alternating paths. Note that it allows us to tighten the previous bound
by subtracting the latencies of the backward edges in the alternating path. The
lemma provides an upper bound on the social cost of the RAWE $\xx$ by
exploiting the equilibrium conditions on the subpaths $B_i$ on the alternating
path with respect to the risk-averse objective. 

\begin{lemma}\label{lemma1}
Consider an arbitrary graph with general latencies
and general variance functions that satisfy that the variance-to-mean ratio at
equilibrium is bounded by $\kappavar$. Letting $\altp$ be an alternating path,
the social cost $\CC(\xx)$ of a RAWE $\xx$ is upper bounded by
$$(1+\gamma\kappavar) \sum_{e\in A\cap \altp}\ell_e(\xx_e) -  \sum_{e\in B\cap \altp}\ell_e(\xx_e)\,.$$
\end{lemma}
\begin{proof}
Let us assume that the alternating path consists of subpaths
$A_1 B_1 A_2 \ldots A_{\eta-1} B_{\eta-1} A_\eta$,
where each subpath is in the corresponding set $A$ or $B$. Since by definition
each edge $e$ in $B_k$ carries flow ($\xx_e>0$) for any $k$, $e$ must belong to a flow-carrying $s$-$t$ path. Selecting a decomposition where the
whole subpath $B_k$ is on the same path (we have the freedom to do that since this
is a standard Wardrop model with respect to the mean-variance objective), there
must be a flow-carrying path that consists of subpaths $C_k B_k D_k$ where
$C_k$ originates at the source node and $D_k$ terminates at the destination
node (see Figure~\ref{fig:alternatingpath} for an illustration). We define
$C_0=D_n=\emptyset$. To simplify notation, only for the proof of this lemma, we are going to refer to the mean-variance cost of subpath $P$ also by
$P = \sum_{e\in P} (\ell_e(\xx_e)+\gamma \var_e(\xx_e) )$.

We next use the equilibrium conditions to derive bounds on $C_k$ and
$D_k$. Since the subpath $C_k B_k$ carries flow and the subpath $C_{k-1} A_k$
is an alternative route between the endpoints of $C_k B_k$, we have that $C_k +
B_k \leq C_{k-1} + A_k$ for all $k$.  Note that here and in what follows we critically use the additivity of the mean-variance cost.
Therefore,
\begin{align}
C_k &\leq C_{k-1} + A_k - B_k \label{ineq:ck}\\
 & \leq C_{k-2} + A_{k-1} + A_k - (B_{k-1} + B_k) \nonumber\\
 & \ldots \nonumber\\
 & \leq (A_1+A_2+\ldots+A_k) - (B_1+B_2+\ldots+B_k) \nonumber
\end{align}
Similarly, since $B_k D_k$ carries flow and $A_{k+1} D_{k+1}$ is an alternative
route between the same endpoints, we have that $B_k + D_k \leq A_{k+1} +
D_{k+1}$ for all $k$.
Therefore,
\begin{align}
D_k &\leq A_{k+1} + D_{k+1} - B_k \label{ineq:dk}\\
 & \leq A_{k+1} + A_{k+2} + D_{k+2} - (B_{k} + B_{k+1}) \nonumber\\
 & \ldots \nonumber\\
 & \leq (A_{k+1}+A_{k+2}+\ldots+A_n) - (B_{k} + B_{k+1} + \ldots+B_{n-1}) \nonumber
\end{align}
Then, for path $q=C_k B_k D_k$ for any $k$, we have that
\begin{align*}
\CC(\xx) = & \sum_{\route} \xx_{\route} \ell_{\route}(\xx) \\
\leq & \sum_{\route} \xx_{\route} (\ell_q(\xx) + \gamma\var_q(\xx) - \gamma\var_{\route}(\xx))
\qquad &&\text{ because either } \xx_{\route} = 0 \mbox{ or } \pathcost_{\route}(\xx) \leq \pathcost_q(\xx) \\
\le & C_k B_k D_k
\qquad &&\text{ after neglecting the negative term}\\
\le & (A_1+\ldots+A_n) - (B_1+ \ldots+B_{n-1})
\qquad &&\text{ using inequalities \eqref{ineq:ck} and \eqref{ineq:dk}}\\
\le&  \sum_{i=1}^n \sum_{e\in A_i} (\ell_{e}(\xx_e)+ \gamma\var_{e}(\xx_e))
	-\sum_{i=1}^{n-1} \sum_{e\in B_i}\ell_{e}(\xx_e)  \qquad &&\text{ neglecting variances in the negative term} \\
\le&  \sum_{i=1}^n \sum_{e\in A_i} (\ell_{e}(\xx_e)+ \gamma \kappavar \ell_{e}(\xx_e))
	-\sum_{i=1}^{n-1} \sum_{e\in B_i}\ell_{e}(\xx_e) \qquad &&\text{ applying the variability bound on the variances} .
\end{align*}
The claim follows.
\end{proof}

The previous result provided an upper bound for the RAWE $\xx$. Now, we complement it
with a lower bound for the RNWE $\zz$. Again, to get the result we exploit the
equilibrium conditions, now with respect to $\ell(\cdot)$.

\begin{lemma}\label{lemma2}
Consider a general instance with a single source-sink pair, general latencies
and general variance functions. Letting $\altp$ be an alternating path,
the social cost of a RNWE $\zz$ satisfies:
$$\CC(\zz)\ge \sum_{e\in A\cap \altp}\ell_e(\zz_e) -  \sum_{e\in B\cap \altp}\ell_e(\zz_e)\,.$$
\end{lemma}
\begin{proof}
%
Since $\zz_e>0$ for any $e\in A_k$, there must be a subpath $C_{k-1}$ that
brings flow to $A_k$ (this $C_{k-1}$ need not be the same as that used in the
proof of Lemma~\ref{lemma1}). Then, there is a flow decomposition in which the
subpath $C_{k-1} A_k$ is used by $\zz$. Because subpath $C_k B_k$ is an
alternative route from $s$ to the node at the end of $A_k$, we must have that
$\ell_{C_{k-1}}(\zz)+\ell_{A_k}(\zz)\le \ell_{C_k}(\zz)+\ell_{B_k}(\zz)$.
Summing the previous inequalities for all $k$ (where $C_0$ is defined as an
empty path), we get
$$\ell_{C_{n-1}}(\zz)\ge \sum_{k=1}^{n-1}(\ell_{A_k}(\zz)-\ell_{B_k}(\zz))\,.$$
This proves the lemma because $\CC(\zz)= \ell_{C_{n-1}}(\zz)+
\ell_{A_{n}}(\zz)$, since $C_{n-1}A_n$ is a flow-carrying path for $\zz$.
\end{proof}

With the previous two lemmas that provided bounds for $\xx$ and $\zz$ and the
sets $A$ and $B$ that allow us to compare both flows, the proof of the main
result consists of just chaining the inequalities.

\begin{theorem}\label{thm:general}
Consider a general instance with a single source-sink pair, general latencies
and general variance functions that satisfy that the variance-to-mean ratio at
equilibrium is bounded by $\kappavar$. Letting $\altp$ be an alternating path,
the price of risk aversion is upper bounded by
$1 + \gamma\kappavar \eta$, where $\eta$ is the number of disjoint forward subpaths in the alternating path $\altp$.
In other words,
$$
\CC(\xx) \leq (1 + \gamma\kappavar \eta ) \CC(\zz).
$$
\end{theorem}
\begin{proof}
\begin{align*}
\CC(\xx) \leq& (1+\gamma\kappavar) \sum_{e\in A\cap \altp}\ell_e(\xx) -  \sum_{e\in B\cap \altp}\ell_e(\xx) \qquad
	&&\text{ by Lemma~\ref{lemma1}}\\
	\leq& (1+\gamma\kappavar) \sum_{e\in A\cap \altp}\ell_e(\zz) -  \sum_{e\in B\cap \altp}\ell_e(\zz)
	&&\text{ by definition of $A$ and $B$}\\
	\leq& \CC(\zz) + \gamma\kappavar\sum_{e\in A\cap \altp}\ell_e(\zz) \qquad &&\text{ by Lemma~\ref{lemma2}}\\
	\leq& \CC(\zz) + \gamma\kappavar \eta \CC(\zz) \\
	=& \left(1 + \gamma\kappavar \eta \right) \CC(\zz)\,.
\end{align*}
Here, we have used that $\sum_{e\in A\cap \altp}\ell_e(\zz) \le \eta
\CC(\zz)$. This holds because for all forward subpaths $A_k \in \altp$, their
edges satisfy $\zz_e > 0$ so $\ell_{A_k}(\zz) \leq \ell_q(\zz) = \CC(\zz)$
for some path $q$ with $\zz_q>0$ that includes the subpath $A_k$. (The latter
holds because any path flow decomposition is valid for the risk-neutral
equilibrium.)
\end{proof}

The coefficient $\eta$, referred to in the introduction, is the maximum
possible number of disjoint forward subpaths. By way of construction, an
alternating path goes through every node at most once and the number of
forward subpaths is maximized when the path consists of alternating forward
and backward edges, for a total of at most $n-1$ edges.  Therefore $\eta
\le \lceil |\altp|/2 \rceil  \le \lceil (n-1)/2 \rceil $.

\begin{corollary}\label{cor:general}
The price of risk aversion in a general graph is upper bounded by
$1 + \gamma\kappavar \lceil (n-1)/2 \rceil$.
\end{corollary}

The bound depends on the three factors that one would expect (risk
aversion, variability and network size), but perhaps unexpectedly does so in
a linear way and for arbitrary delay and variance functions.
Applying Theorem~\ref{thm:general} to the Braess graph used in
Example~\ref{ex:braessVar} (see Figure~\ref{fig:pra-counterexample}), we can
see that for that family the bound provided by the main theorem is tight.

\begin{corollary}
The price of risk aversion among all instances whose topology is a Braess graph
is exactly $1 + 2\gamma\kappavar$.
\end{corollary}
\begin{proof}
The proof follows directly noting that $n = 4$ for the Braess graph.  Here,
we give more details on the alternating path and resulting PRA, to give more
insight into the problem. There are four possibilities for an alternating
path $\altp$: the three (regular) $s$-$t$ paths and the alternating path
consisting of two forward edges ($b$ and $c$ in the example) and one backward
edge ($e$ in the example). The first cases consist of forward edges only, so
when any of those paths is in $A$, the bound for the price of risk aversion
would be $1+\gamma\kappavar$. The `bad case' is when $\altp$ is the
alternating path with one backward edge. In that case, the alternating path
has two non-adjacent forward edges, providing the matching upper bound to the
example. This implies that the value of the PRA is exact.
\end{proof}

\begin{figure}[t]
\centerline{
\includegraphics[width=2.5in]{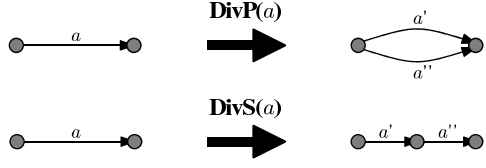}
}
\caption{Operations used to construct a series-parallel graph.}
\label{fig:serpar}
\end{figure}

Next, we derive that the price of
risk aversion in series-parallel graphs is at most $1+\gamma\kappavar$,
independently of the size of the network. Given
the lower bound provided by Example~\ref{ex:pigou} (a Pigou graph is
series-parallel), this bound must be tight. Series-parallel graphs are those
formed recursively by subdividing an edge in two subedges, or replacing an
edge by two parallel edges (see Figure~\ref{fig:serpar}). A noteworthy
alternative characterization is that a graph is series-parallel if and only
if it does not contain a Braess subgraph as an induced minor.

\begin{corollary}\label{cor:sp}
The price of risk aversion among all series-parallel instances
is exactly $1 + \gamma\kappavar$.
\end{corollary}
\begin{proof}
We are going to prove that there exists an alternating path $\altp$ consisting only
of forward edges, so $\altp\subseteq A$. Let us consider a minimal
(cardinality-wise) alternating path with a backward edge. The key property of
series-parallel graphs is that after taking a reverse edge $e^-$, where
$e=(i,j)\in\GA$, $\altp$ has to either come back to node $j$ or close a loop with
itself. If that did not happen, it would imply that a Braess graph is embdeded
in the instance, which is not possible. Hence, there is an alternating path
$\altp'$ without the reverse edge $e^-$, which is a contradiction to the minimality
of~$\altp$.
\end{proof}

\section{Representing Risk as the Standard Deviation}

We now consider a related risk measure based on the standard deviation rather
than the variance. The objective that each user seeks to minimize is a linear
combination of the expectation and the standard deviation along a route.  We
call that the mean-standard deviation ({\em mean-stdev}) objective. For
simplicity, in the rest of this section we refer to the mean-stdev objective as
the risk-averse cost along a route. Formally, the mean-stdev cost along route
$\route$ is
\begin{equation}\label{eqn:pathcost-stdev}
\pathcost^{\gamma}_{\route}(\flow) =\sum_{e\in\route} {\ell_e (f_e)} +\gamma
\sqrt{\sum_{e\in\route} {\std_e(f_e)^2}}\,.
\end{equation}
For this objective, equilibrium existence follows from a variational inequality
formulation if standard deviation functions $\std(x)$ are continuous, as we have assumed here
\cite{nikolova-stochWE}.

Example~\ref{ex:pigou} and Remark~\ref{remark:pigou} for the mean-var model can
be adapted here, replacing the variances with standard
deviations in the example specification. Since for arbitrary instances the PRA
is unbounded, we assume that $\std_e(\xx_e)/\ell_e(\xx_e)$ is no more than a fixed
constant $\kappastdev$ for all $e\in \GA$ at the RAWE $\xx_e\in \R_+$ (this is
less restrictive than requiring such a bound for all feasible flows). This
means that the standard deviation cannot be larger than $\kappastdev$ times the
expected latency in any edge at the equilibrium flow.

We start by identifying which results extend from the mean-var model to the
mean-stdev model here.   Essentially all lemmas extend, except for
Lemma~\ref{lemma1}.  Proving this lemma is thus the only remaining roadblock to
proving the equivalent of Theorem~\ref{thm:general} in the case of the
mean-stdev cost, namely establishing a price of risk aversion bound for general
graphs.

For completeness, we restate the lemmas and some of the proofs that require a slight modification.
%
By the definition of equilibrium, the cost $\CC(\zz)$ of a RNWE can be
bounded by the latency $\ell_\route(\zz)$ of an arbitrary path $\route$, and both
are equal if $\zz_\route>0$. We now extend that argument to a RAWE $\xx$.
We prove that its total cost is bounded by the expected latency of an
arbitrary path, blown up by a constant that depends on the risk-aversion
coefficient and the maximum coefficient of variation.

\begin{lemma}\label{lemma-UBcostRAWE-stdev}
Consider a general instance with a single source-sink pair and general
latencies and standard deviation functions. Letting $\route\in\paths$ denote an
arbitrary path (potentially not carrying flow at equilibrium), the social cost
of a RAWE $\CC(\xx)$ is bounded by the path cost $\pathcost_\route(\xx)$.
\end{lemma}


\begin{corollary}\label{cor-UBcostRAWE-stdev}
Consider a general instance with a single source-sink pair, general latencies
and general standard deviation functions that satisfy that the coefficient of
variation at equilibrium is bounded by $\kappastdev$. Letting $\route\in\paths$
denote an arbitrary path (potentially not carrying flow at equilibrium), the
social cost $\CC(\xx)$ of a RAWE $\xx$ is bounded by
$(1+\gamma\kappastdev)\ell_{\route}(\xx)$.
\end{corollary}
\begin{proof}
From Lemma~\ref{lemma-UBcostRAWE-stdev},
\begin{align*}
\CC(\xx)
\le \pathcost_\route(\xx)
=& \ell_\route(\xx) + \gamma  \sqrt{ \sum_{a \in \route}\std_a^2(\xx) } \\
\leq& \ell_\route(\xx) + \gamma  \sum_{a \in \route}\std_a(\xx) \\
\leq& \ell_\route(\xx) + \gamma  \sum_{a \in \route}\kappastdev\ell_a(\xx)
= \ell_\route(\xx) (1+\gamma \kappastdev)\,.
\end{align*}
Here, we have used
that $||x||_2\le ||x||_1$ for an arbitrary nonnegative vector, and applied the bound of
$\kappastdev$ on the coefficient of variation.
\end{proof}

As in the mean-var case, to get the tightest possible upper bound, one would consider the path with smallest standard deviation induced by the RAWE. Selecting that path, we can get rid of the factor and bound the total cost by the expected latency of the path.

\begin{lemma}\label{thm-UBcostRAWE2-stdev}
Consider a general instance with a single source-sink pair, general latencies
and general standard deviation functions. Letting $\route\in\paths$ denote the path that
minimizes the standard deviation under a RAWE $\xx$ (where the path $\route$ may or may not
carry flow), the social cost $\CC(\xx)$ is bounded by $\ell_{\route}(\xx)$.
\end{lemma}
Let us note that an alternating path exists using the same proof as that of Lemma~\ref{lemma0}, even though the corresponding sets $A$ and $B$ will be different reflecting that $\xx$ is now a RAWE for the mean-stdev cost.
We now prove a special case of Theorem~\ref{thm:general} for the mean-stdev model, namely that if we can find an alternating path in the set $A$, then PRA in the mean-stdev model is $1+\gamma\kappastdev$.

\begin{theorem}\label{thm:stdev}
Consider a general instance with a single source-sink pair, general latencies
and general standard deviation functions that satisfy that the stdev-to-mean
ratio at equilibrium is bounded by $\kappastdev$.  Suppose there exists an
alternating path $\route$ consisting of forward edges only (namely edges in the set
$A = \{e\in \GA \,|\, \zz_e \geq \xx_e \mbox{ and } \zz_e > 0 \} $.
Then, the price of risk aversion is upper bounded by
$1 + \gamma\kappastdev$.
In other words,
$$
\CC(\xx) \leq \left(1 + \gamma\kappastdev\right) \CC(\zz).
$$
\end{theorem}
\begin{proof}
Since $\route\subseteq A$ is composed of only forward edges, it is an $s$-$t$ path,
which allows us to apply Corollary~\ref{cor-UBcostRAWE-stdev}.
Lemma~\ref{lemma2} also holds since it concerns the RNWE $\zz$, which does not
depend on how risk aversion is captured by the model.
The rest of the proof follows as before:
\begin{align*}
\CC(\xx) \leq& (1+\gamma\kappastdev) \sum_{e\in \route}\ell_e(\xx_e)  \quad
	&&\text{ by Corollary~\ref{cor-UBcostRAWE-stdev}}\\
	\leq& (1+\gamma\kappastdev) \sum_{e\in \route}\ell_e(\zz_e)
	&&\text{ by monotonicity of the latency functions, and since path $\route \subseteq A$}\\
	\leq& \left(1 + \gamma\kappastdev\right) \CC(\zz) \qquad &&\text{ by Lemma~\ref{lemma2}} \,.
\end{align*}
\end{proof}

Extending Lemma~\ref{lemma1} to the mean-stdev objective for general graphs
remains elusive. The proof for the mean-var objective relies on the equilibrium
conditions on subpaths of the RAWE $\xx$. Although the mean-var objective leads to a
separable model, the mean-stdev one does not (for details, we refer the reader
to \cite{nikolova-stochWE}), and that complicates a general proof. Moreover, as
an additional complication, in the variance case we use a flow decomposition
that suits our needs but for the standard deviation case, the decompositions
cannot be arbitrary \cite{nikolova-stochWE} so we cannot guarantee that the one we need is valid. As
some preliminary steps to a general proof, the next two sections provide tight
bounds for PRA in two well-studied families of graphs: Braess networks and
series-parallel networks.

\subsection{PRA in the Braess Network}\label{sec:braess-mean-var}

In this section we consider Braess paradox networks (see
Figure~\ref{fig:pra-counterexample}) with a unit demand.
We use the same notation as within Example~\ref{ex:braessVar}.


%

We will show that a slight modification of our proof for the mean-variance
cost function in general graphs can provide a proof for the mean-stdev
cost function in the Braess graph, and we comment on the challenge of
extending such a proof to a mean-stdev cost function in general graphs.
%
In what follows, we will prove a version of Lemma~\ref{lemma1} for the
mean-stdev in the Braess graph and consequently we will get that the PRA is
$1+2\gamma\kappastdev$. We start with an auxiliary lemma. As before, we
denote the top path by $p$, consisting of edges $(a,b)$; the bottom path by
$q$, consisting of edges $(c,d)$, and the zigzag path by $r$, consisting of
edges $(a,e,d)$.

\begin{lemma}\label{lemmaBraess}
For an arbitrary flow $f$ in a Braess network,
$\std_\braessp(f)+ \std_\braessq(f)-\std_\braessr(f)\le \std_b(f_b)+\std_c(f_c)$ when $\std_\braessr(f)\le \max(\std_\braessp(f),\std_\braessq(f))$.
\end{lemma}
\begin{proof}
To simplify notation and since the flow $f$ does not play a role in the proof, we are going to suppress the dependence on $f$ of the standard deviation functions.
The inequality we wish to prove is equivalent to
$$
\sqrt{\std_a^2+\std_b^2}+ \sqrt{\std_c^2+\std_d^2} \le \std_b+\std_c+\sqrt{\std_a^2+\std_e^2+\std_d^2}\,.$$
Squaring and rearranging terms, it is also equivalent to
$$2\sqrt{(\std_a^2+\std_b^2)(\std_c^2+\std_d^2)}\le 2\std_b\std_c+\std_e^2+2(\std_b+\std_c)\sqrt{\std_a^2+\std_e^2+\std_d^2}.\,$$
Finally, squaring once more, we get
\begin{multline*}
\std_a^2\std_d^2\le \frac{\std_e^4}{4}+\std_b\std_c\std_e^2+\std_a^2\std_b^2+\std_c^2\std_d^2+2\std_b\std_c(\std_a^2+\std_e^2+\std_d^2)+\nonumber\\
\quad \std_e^2(\std_b^2+\std_c^2)+(2\std_b\std_c+\std_e^2)(\std_b+\std_c)\sqrt{\std_a^2+\std_e^2+\std_d^2}\,.
\end{multline*}
If $\std_\braessp\ge\std_\braessr$, the last inequality holds because
$\std_a^2+\std_b^2\ge \std_a^2+\std_e^2+\std_d^2$, from where
$\std_b\ge \std_d$. The case of $\std_\braessq\ge\std_\braessr$ is similar.
\end{proof}

We now prove a variant of Lemma~\ref{lemma1} for the mean-stdev cost on Braess
graphs.

\begin{lemma}\label{lemma1:braess-stdev}
Consider a Braess network, general latencies and general standard deviation
functions that satisfy that the coefficient of variation at equilibrium is
bounded by $\kappastdev$. Letting $\altp$ be an alternating path,
the social cost $\CC(\xx)$ of a RAWE $\xx$ is upper bounded by
$$(1+\gamma\kappastdev) \sum_{e\in A\cap \altp}\ell_e(\xx_e) -  \sum_{e\in B\cap \altp}\ell_e(\xx_e)\,.$$
\end{lemma}
\begin{proof}
If $\altp\subseteq A$, the result follows from Theorem~\ref{thm:stdev}. This
can happen when $\altp\in\{\braessp,\braessq,\braessr\}$. Otherwise $\altp$
is the alternating path that consists of edges $c$, $e^-$, and $b$. In that
case, $e\in B$.  Note that $\zz_e=\xx_e=0$ cannot hold for $e$ because it is part
of the alternating path, so we must have that $\xx_e>0$ and consequently that
$\xx_\braessr>0$. We must prove that
$$\CC(\xx)\le (1+\gamma\kappastdev) (\ell_c(\xx_c)+\ell_b(\xx_b))-\ell_e(\xx_e)\,.$$
Since $\xx_\braessr > 0$, we have that
\begin{eqnarray}
&& \ell_\braessr(\xx) + \gamma \std_\braessr(\xx)\leq \ell_\braessp(\xx) + \gamma \std_\braessp(\xx)\nonumber \\
\Leftrightarrow && \ell_a(\xx) + \ell_e(\xx) + \ell_d(\xx) \leq \ell_a(\xx) + \ell_b(\xx) + \gamma (\std_\braessp(\xx)-\std_\braessr(\xx)) \nonumber \\
\Leftrightarrow &&  \ell_d(\xx) \leq \ell_b(\xx) -\ell_e(\xx) + \gamma (\std_\braessp(\xx)-\std_\braessr(\xx)). \label{eq:d}
\end{eqnarray}
Similarly,
\begin{eqnarray}
&& \ell_\braessr(\xx) + \gamma \std_\braessr(\xx)\leq \ell_\braessq(\xx) + \gamma \std_\braessq(\xx)\nonumber \\
\Leftrightarrow && \ell_a(\xx) + \ell_e(\xx) + \ell_d(\xx) \leq \ell_c(\xx) + \ell_d(\xx) + \gamma (\std_\braessq(\xx)-\std_\braessr(\xx)) \nonumber \\
\Leftrightarrow &&  \ell_a(\xx) \leq \ell_c(\xx) -\ell_e(\xx) + \gamma (\std_\braessq(\xx)-\std_\braessr(\xx)). \label{eq:a}
\end{eqnarray}

If $\std_\braessr(\xx)\le \max\{\std_\braessp(\xx),\std_\braessq(\xx)\}$, then we can apply Lemma~\ref{lemmaBraess}:
\begin{align*}
\CC(\xx)
\leq & \ell_\braessr(\xx) + \gamma\std_\braessr(\xx)\qquad\text{ by Lemma~\ref{lemma-UBcostRAWE-stdev}} \\
= & \ell_a(\xx) + \ell_e(\xx) + \ell_d(\xx) +\gamma\std_\braessr(\xx)\\
\leq & \ell_c(\xx) -\ell_e(\xx)+ \gamma(\std_\braessq(\xx)-\std_\braessr(\xx))+ \ell_e(\xx) + \ell_b(\xx) -\ell_e(\xx)+ \gamma(\std_\braessp(\xx)-\std_\braessr(\xx)) +\gamma\std_\braessr(\xx)\qquad\text{ by \eqref{eq:d}-\eqref{eq:a}}\\
= & \ell_c(\xx) + \ell_b(\xx) -\ell_e(\xx)+\gamma(\std_\braessp(\xx)+ \std_\braessq(\xx)-\std_\braessr(\xx))\\
\le & \ell_c(\xx) + \ell_b(\xx) -\ell_e(\xx)+\gamma(\std_b(\xx_b)+\std_c(\xx_c)) \qquad\text{ by Lemma~\ref{lemmaBraess}}\\
\le & (1+\gamma\kappastdev) (\ell_c(\xx_c)+\ell_b(\xx_b))-\ell_e(\xx_e),
\end{align*}
where the last inequality follows from the coefficient of variability constraint  $\std(\xx) \leq \kappastdev \ell(\xx)$ for the corresponding edges.

Otherwise, $\std_\braessr(\xx)> \std_\braessp(\xx)$ and
$\std_\braessr(\xx)> \std_\braessq(\xx)$.
Then, inequalities~\eqref{eq:d}
and \eqref{eq:a} imply that
\begin{align*}
\ell_b(\xx)&\geq \ell_d(\xx) +\ell_e(\xx) + \gamma (\std_\braessr(\xx)-\std_\braessp(\xx)) \geq \ell_d(\xx) +\ell_e(\xx)\text{ and}\\
\ell_c(\xx)&\geq \ell_a(\xx) +\ell_e(\xx) + \gamma (\std_\braessr(\xx)-\std_\braessq(\xx)) \geq \ell_a(\xx) +\ell_e(\xx).
\end{align*}
Summing the above two inequalities, we have:
$
\ell_b(\xx)+\ell_c(\xx)\geq \ell_d(\xx) +\ell_e(\xx) + \ell_a(\xx) +\ell_e(\xx) = \ell_\braessr(\xx)+\ell_e(\xx)
$.
Therefore,
\begin{align*}
\CC(\xx) &\leq (1+\gamma\kappastdev)\ell_\braessr(\xx) \qquad\text{ by Corollary~\ref{cor-UBcostRAWE-stdev}} \\
&\leq (1+\gamma\kappastdev)\left(\ell_b(\xx)+\ell_c(\xx)-\ell_e(\xx)\right) \qquad\text{ by inequality above}\\
&\leq (1+\gamma\kappastdev)\left(\ell_b(x)+\ell_c(\xx)\right)-\ell_e(\xx).
\end{align*}
This completes the proof of this lemma.
\end{proof}

Using Lemma~\ref{lemma1:braess-stdev} in place of Lemma~\ref{lemma1}, we can
apply the proof of Theorem~\ref{thm:general} with the alternating path
$\altp$. If $\altp$ is a path, the PRA is bounded by $1 + \gamma\kappastdev$.
Otherwise, $\altp$ consists of edges $b,e,c$ where $\{b,c\}\subseteq A$ and
$e\in B$ and the PRA is bounded by $1 + \gamma\kappastdev \lceil (n-1)/2
\rceil  = 1 + 2\gamma\kappastdev$.

The extension to general networks seems much harder than in the case where the cost is the mean-var objective because not all flow decompositions are valid.

\subsection{Series-Parallel Networks with Single OD and General Latencies}

In this section, we consider series-parallel (SP) networks and provide an
analogous result to Corollary~\ref{cor:sp} for the mean-stdev case, 
namely that the PRA is $(1+\gamma\kappastdev)$. The
extension is straightforward with the tools we have established since SP
networks always possess an alternating path that is a subset of $A$ (see proof
of Corollary~\ref{cor:sp}) and hence we can apply Theorem~\ref{thm:stdev} to
get the following result.

\begin{corollary}\label{cor:spStdev}
Consider a SP graph with a single source-sink pair, general latencies
and general standard deviations with coefficients of variation bounded by
$\kappastdev$. Then, the PRA is bounded by $1+\gamma \kappastdev$.
\end{corollary}

The tightness of this bound follows after adapting Example~\ref{ex:pigou} to
the case of mean-stdev costs.

In the rest of this section, we offer an alternative proof that is of independent interest.
This proof sheds light on the inefficiency of Wardrop equilibria, which
has been studied at length in the last 15 years, as we discuss in the
introduction. Our proof establishes that in a SP network the risk-neutral equilibrium $\zz$
(equivalently, a standard Wardrop equilibrium) maximizes the shortest-path length among all feasible flows. This implies that at a risk-neutral equilibrium travelers select paths that are longer in
expectation than the shortest path achieved under any other flow.

The following result provides an upper bound to the PRA proportional on the
worst-case ratio between the shortest expected path latency under a RAWE and that
under a RNWE. We remark that this bound also holds in the mean-var case.

\begin{theorem}\label{ub:rho}
Consider an instance with a single source-sink pair, general latencies
and general standard deviations with coefficients of variation bounded by
$\kappastdev$. Then, the PRA is bounded by $(1+\gamma \kappastdev )\rho$, where 
$\rho:= \min_\route \ell_\route(\xx)/ \min_q \ell_q(\zz)$ for a RAWE $\xx$ and a RNWE $\zz$.
\end{theorem}
\begin{proof}
Using Corollary~\ref{cor-UBcostRAWE-stdev}, we have that $C(\xx)/C(\zz)\le
(1+\gamma \kappastdev ) \ell_\route(\xx)/ \min_q \ell_q(\zz)$ for an arbitrary path
$\route$. We get the result by minimizing the expected latency over paths $\route\in
\paths$.
\end{proof}
%


The alternative proof consists in showing that $\rho\le 1$ for SP networks, or
$\min_\route \ell_\route(\xx)\le \min_\route \ell_\route(\zz)$. To simplify
notation, we refer to the shortest path with respect to expected latencies
corresponding to a feasible flow $f$ by $S(f):= \min_{\route\in\paths} \ell_\route(f)$.
To get the result we now prove that $S(f)\le S(\zz)$ for any feasible flow $f$.

\begin{theorem}\label{thm:SP}
Consider a RNWE $\zz$ of a SP network with general latencies and one source-sink pair. Then,
$\zz$ maximizes $S(f)$ among all feasible flows $f$.
\end{theorem}
\begin{proof}
We use induction on the construction steps of the SP network. 
First, let us consider that the last composition in the construction of the
network is series. Because it is a series composition, the restriction of $\zz$
to each component is a RNWE for the component. Considering an arbitrary
feasible flow $f$, the induction implies that the restriction of $f$ to each
component cannot have a longer shortest path than the restriction of $\zz$ in
the same component. We get the desired inequality adding the inequalities
corresponding to each component back together.

Second, let us consider that the last composition in the construction of the
network is parallel and denote the subcomponents by $G_1, G_2, \ldots, G_k$.
Let us denote by $\zz_i$ and $f_i$ the projection of each flow into a
component. There must exist a component $i$ such that $f_i\le \zz_i$ and
$\zz_i>0$ because $\sum \zz_i=\sum f_i=1$. Then, $S(f)\le S(f_i)\le
S(RNWE_{G_i}(f_i)) \le S(RNWE_{G_i}(\zz_i))=S(\zz)$. Here, we have denoted by
$RNWE_{G_i}(d)$ the risk-neutral equilibrium in subgraph $G_i$ for demand $d$.
The first inequality is because $S(f)$ is the minimal shortest path across
components, the second is by the inductive hypothesis applied to the component
$G_i$ with demand $f_i$, the third is because increasing demand from $f_i$ to
$\zz_i$ cannot reduce the shortest path at equilibrium, and the fourth is
because some flow is routed through component $i$ so the shortest path in that
component is equal to the shortest path in the whole graph.
\end{proof}

This argument does not readily extend to non-SP networks since we can sometimes route flow in a worse way
than a Wardrop equilibrium.
As an example, take the instance shown in Figure~\ref{fig:bad} with
$k$ horizontal edges and $d=1$. Although the equilibrium loads all horizontal
paths equally achieving a shortest path length of $1/k$, to maximize the
shortest path one can route all flow along the zigzag path achieving a
shortest path of 1. 
\begin{figure}
\centering\includegraphics[width=2.5in]{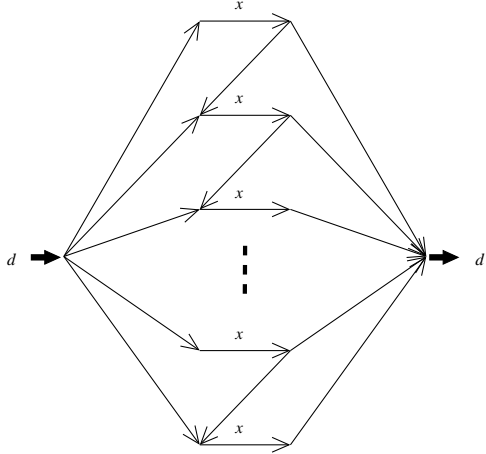}
\caption{A non-SP example where the RNWE does not maximize the shortest path.}\label{fig:bad}
\end{figure}
In any case, note that the example above does not preclude the possibility
that the PRA is bounded for non-SP networks. Although $S(f)$ could be high
for an arbitrary feasible flow $f$, it does not imply that $S(\xx) > S(\zz)$. Indeed
$\xx$ is unlikely to use very bad paths because it is at equilibrium
for a different objective function, whereas the bad flow in the zigzag-path
example is not an equilibrium for any objective.

\section{Conclusion}

%
%
%
%
%

This paper marks a first step in understanding the consequences on the inefficiency of selfish
routing caused by uncertain edge delays and risk-averse players. We have established an upper
bound on the ratio of the cost of the risk-averse equilibrium to that of the
risk-neutral one, for users that aim to minimize the mean-variance of their route in a
general network. We have proved that the bound is tight on series-parallel
and Braess networks.  In addition, we have shown that both tight bounds
extend to the case where users minimize the mean-standard deviation of a route instead
and have elaborated on
the challenges of extending the analysis to general graphs for users with such
risk profiles. Some immediate open questions include: (1) Is our bound for
general graphs tight? (2) Can it be extended to mean-stdev and other risk
objectives?  (3) Can the bounds or analysis be extended to heterogeneous risk
profiles? Another interesting direction is whether risk sometimes helps rather
than hurts the quality of equilibrium and the social welfare.  In particular,
can risk-averse attitudes be leveraged in mechanism design in the place of
tolls to reduce congestion?

\section{Appendix: A proof of Theorem~\ref{thm:SP} via KKT conditions}

Without loss of generality we may consider that the instance has unit demand.
Let us consider the convex optimization problem that maximizes the shortest
path objective among feasible flows. Our approach consists on proving that
the optimal objective value is $S(\zz)$, which by definition of equilibria
equals $C(\zz)$. In general, the optimal objective value cannot be less than
$C(\zz)$ because $\zz$ is a feasible solution to problem below. Although for
general networks the optimum can be strictly larger, we will see that this
cannot happen for SP networks. In the following formulation, we denote the
dual variables within parenthesis.

%

\begin{align*}
\min \qquad&z&\\
&-z- \sum_{a\in p} \ell_a (f_a)\le 0 \qquad &\forall p\in\paths \qquad &(\lambda_p)\\
&-f_p\le 0\qquad &\forall p\in\paths \qquad &(\omega_p)\\
&\sum_{p\in\paths} f_p-1=0 \qquad &\qquad&(\mu)\\
&f_a-\sum_{p:a\in p\in\paths} f_p =0  \qquad &\forall a\in A \qquad &(\nu_a)\\
\end{align*}

The problem is written in normal form to facilitate writing the dual. According
to the first constraint, $-z$ is a lower bound to the length of all paths,
hence correctly capturing the shortest-path objective. The second constraint
imposes non-negativity of the flow, the third specifies the total demand and
the last defines the flows on edges.
The Lagrangian dual is
$$\Lambda(z,f,\lambda,\omega,\nu,\mu):=z- \sum_{p\in\paths } \lambda_p(z+ \sum_{a\in
p} \ell_a (f_a))- \sum_{p\in\paths } \omega_p f_p+ \mu(\sum_{p\in\paths } f_p-1)+ \sum_{a
\in A} \nu_a(f_a-\sum_{p:a\in p\in\paths} f_p).$$
The stationarity conditions together with dual feasibility are given by the
domain
\begin{align*}
&\sum_{a\in p} \lambda_a \ell'_a(f_a) \le \mu \qquad &\forall p\in\paths \qquad &(f_p)\\
&\sum_{p\in\paths} \lambda_p=1 \qquad &\qquad&(z)\\
&\lambda_a=\sum_{p:a\in p\in\paths} \lambda_p   \qquad &\forall a\in A \qquad &\\
&\lambda_p\ge 0, \mu \text{ free}\\
\end{align*}
and the complementary slackness conditions are given by
\begin{align*}
&\lambda_p(\sum_{a\in p} \ell_a (f_a)-z)=0&\forall p\in\paths \qquad &\\
&f_p(\mu-\sum_{a\in p} \lambda_a \ell'_a(f_a))=0&\forall p\in\paths\,. \qquad &\\
\end{align*}

We find feasible values for $z$, $f$, $\lambda$, and $\mu$ that satisfy
primal and dual feasibility, the stationary and the complementary slackness
conditions at the same time. For that purpose, we take $f=\zz$, which is
primal feasible. For the dual, we consider the function $\zz(d)$ that maps a
total demand $d$ to the RNWE edge-flow, which is unique when cost functions
are strictly increasing. Using this, we let $\lambda_p$ be the marginal
increase in flow along path $p$ when we increase the demand by an
infinitesimal. Formally, we let $\lambda_a:=\partial \zz_a(1)/\partial d$ and
consider an arbitrary decomposition of the edge flow $(\lambda_a)_{a\in A}$ to
get $\lambda_p$ for all $p\in\paths$. This definition automatically gives us
dual feasibility because $\sum_{p\in\paths} \lambda_p=1$. This happens
because the additional flow has to be assigned to some path.

What is left is proving that $(f,\lambda)$ is an optimal primal-dual pair by
checking that they satisfy the complementary-slackness conditions. The first
condition holds because a path that is not shortest under a RNWE will not
receive flow if the demand increases by an infinitesimal.

For the second condition, we need to prove that if the RNWE sends flow along
a path $p$ then $\sum_{a\in p} \lambda_a \ell'_a(f_a)$ must be maximal among
paths. Let $q$ be a path such that the expression in the previous sentence
equals $\mu$ and that additionally has the lowest number of edges with
$\lambda_a=0$. If $\lambda_a=0$ for all $a\in q$, then $\mu=0$ and we are
done. Hence, there is at least one $a\in q$ such that $\lambda_a>0$. If
there are other edges $i\in \tilde q\subset q$ with $\lambda_i=0$, we consider
a path $r$ such that $\lambda_i>0$ for all $i\in r$ that maximizes the
overlap with $q$. (Obviously, $r\cap\tilde q=\emptyset$.) Using that the
network is series-parallel, we notice that $r$ shortcuts all the subpaths in
$\tilde q$ and hence $\sum_{a\in r} \lambda_a \ell'_a(f_a)=\mu$ because the
terms corresponding to $\tilde q$ are zero.

The argument in the previous paragraph allows us to assume that $\lambda_r>0$
by taking the correct path-decomposition of $(\lambda_a)_{a \in A}$. Since
increasing the total demand increases the flow along $r$, the cost at
equilibrium along $r$ must be shortest (even though $f_r$ might be zero).
Coming back to the second condition, for any path $p\in\paths$ such that
$f_p>0$, the cost at equilibrium along $p$ must also be shortest and hence
the additional load along both $r$ and $p$ when increasing demand must be
equal to $\mu$.

\bibliographystyle{plainnat}
\bibliography{../../rwe,../../../soue,../../../routeguid,../../../telecom,../../eddie-thesis,../../risk,../../risk2,../../../robust} 



\end{document}